\documentclass[12pt]{amsart}
\usepackage{amsfonts,amssymb,graphicx,mystyle,color}
\usepackage[colorlinks,linkcolor=blue]{hyperref} 
\usepackage{cite}
\usepackage{hyperref}
 \let\backslash=\setminus \let\ge=\geqslant \let\le=\leqslant 
  \def\d{\delta} \def\e{\varepsilon}    \def\s{\sigma} \def\t{\tau} 
   
\def\S{\Sigma}
     \def\N{\mathcal N}

\def\Re{\mathbb{R}}

\def\k{\kappa}
\def\rho{\varrho}

\def\eL{\mathcal L}

\begin{document}\openup 1\jot
\title[Perfect Equilibria]{A Finite Characterization of Perfect Equilibria \\}
\author[I. Callejas]{Ivonne Callejas}
\address{Mathematisches Institut, Georg August Universit\"{a}t G\"{o}ttingen, Bunsenstr.3-5, 37073 \newline
\indent  G\"{o}ttingen, Germany.}
\email{icallej@mathematik.uni-goettingen.de}
\author[S. Govindan]{Srihari Govindan}
 \address{Department of Economics, University of Rochester, NY 14627, USA.}
 \email{s.govindan@rochester.edu}
 \author[L. Pahl]{Lucas Pahl}\thanks{We are grateful to Saugata Basu, Goulwen Fichou, and Heng Liu for bringing the work of Kurdyka and Spodzieja \cite{KS2014} as well as Basu and Roy \cite{BR2018} to our attention; we would also like to thank Rida Laraki, Sven Rady, the participants of the BSGE Micro Theory Workshop and the One World Game Theory Seminar for their comments and suggestions.}
 \address{Institute for Microeconomics, University of Bonn, Adenauerallee 24-42, 53113 Bonn, Germany.}
 \email{pahl.lucas@gmail.com}
 \date{August, 2020; this version: August, 2021}
\maketitle

\begin{abstract}
Govindan and Klumpp \cite{GK2002} provided a characterization of perfect equilibria using Lexicographic Probability Systems (LPSs). Their characterization was essentially finite in that they showed that there exists a finite bound on the number of levels in the LPS, but they did not compute it explicitly.  In this note, we draw on two recent developments in Real Algebraic Geometry to obtain a formula for this bound. \end{abstract}

\section{Introduction}

The concept of perfect equilibrium \cite{S1975} plays a central role in the theory of refinements of Nash equilibria. Not only has it been successful in applications to economic models, but its pioneering use of trembles has spurred further refinements.  From a practical  viewpoint, the very nature of its definition makes perfect equilibrium very hard to compute, which explains the relative paucity of algorithms to compute it.\footnote{A recent advance in computing perfect equilibria is the homotopy method in \cite{CD2019}.}   Indeed,  recall that an equilibrium of a finite game in normal form is perfect if there exists a sequence of profiles of completely-mixed strategies converging to it against which the equilibrium is a best reply.  The set of perfect equilibria of a finite game is, thus, defined by finitely many polynomial inequalities as well the universal ($\forall$) and existential ($\exists$) quantifiers. Consequently, checking whether a strategy profile is perfect, let alone computing the entire set, requires, in principle, an infinite number of operations, as we have to contend with the quantifiers: for every $\e > 0$, there exists $\d > 0$ such that...  However, the set of perfect equilibria is a semi-algebraic set---see \cite{BZ1994}.   Therefore, by the Tarski-Seidenberg Theorem \cite{BCR1998}, there exists an equivalent description of perfect equilibria that is quantifier-free.  It would be really convenient, then,  to know what such a description would look like. This question is open and seems hard to resolve.  A more modest question is whether we can eliminate the universal quantifier;  in other words, we are after an equivalent definition  of perfection of the following form: a strategy profile is perfect if there exists a solution to a finite system of polynomial equations and inequalities (in which the given strategy profile is a parameter).  Such a formulation provides a finite characterization of perfection, since it requires only a finite number of steps to check whether a finite system of polynomial equations and inequalities has a solution.  Govindan and Klumpp \cite{GK2002} (henceforth GK) obtain a result of this kind.  The system that they consider is shown to have a bound on the number of polynomials involved; this bound depends \textit{only} on the cardinalities of the player set and the strategy sets of the players, but it is not computed.  In this note, we sharpen the result in GK by giving explicit bounds. It is worth noting that the polynomial system in this characterization involves only multilinear polynomials, which are especially suitable for the application of polynomial-solving algorithms (see \cite{DA2010} and \cite{MM1996}).



What make this paper possible are recent advances in quantifying the bounds involved in two results from Real Algebraic Geometry that were used by GK.  The first concerns what is called the \L{}ojasiewicz inequality, which states that the value of a polynomial in a neighborhood of a zero is bounded from below by a polynomial function of the distance from the zero-set of the polynomial. Kurdyka and Spodzieja \cite{KS2014} give an explicit formula for  the degree of the bounding polynomial. The second concerns the Nash curve-selection lemma \cite{BCR1998}. Suppose $X$ is a semi-algebraic set and $x$ belongs to its closure.  Then, the curve-selection lemma states that there is an analytic function from an interval $[0, \e]$ into the closure of $X$ that maps $0$ to $x$ and maps all other points into $X$.  Basu and Roy \cite{BR2018} provide a quantitative version of this lemma that gives us a bound on the degree of the (coordinate) analytic functions.

\section{Definitions and Statement of The Theorem}\label{sec definitions}

We study a finite game $G$ in normal form.  The set of players is $\N = \{\, 1, \ldots, N \, \}$. The finite set of pure strategies of each player $n \in \N$ is denoted $S_n$, and the corresponding set of  mixed strategies is $\S_n$.  Define $S \equiv \prod_n S_n$ and $\S \equiv \prod_n \S_n$. For each $n$, $S_{-n} = \prod_{m \neq n} S_m$;  $\S_{-n} \equiv \prod_{m \neq n} \S_m$. The payoff function of player $n$ is $G_n: \S \to \Re$.

\begin{definition}\label{def perfection}
The profile $\s \in \S$ is a perfect equilibrium if there exists a  sequence $\s^k$  of profiles of completely-mixed strategies converging to $\s$ such that for each $k$, each player $n \in \N$, and each strategy $s_n \in S_n$, $G_n(\s_n, \s_{-n}^k) \ge G_n(s_n,\s_{-n}^k)$. 
\end{definition}

GK provide an equivalent definition of perfect equilibria that replaces the test sequence $\s^k$ with a lexicographic probability system (LPS) for each player.\footnote{Blume et al \cite{BBD1991} provide a characterization in terms of LPSs defined on the set $S$ of pure-strategy profiles.  But, their definition is not semialgebraic: in particular, one of their conditions---called strong independence---cannot be verified in finitely many steps.}  To describe their characterization, we need some definitions.

\begin{definition}\label{def LPS}
Let $K$ be a non-negative integer. An LPS of order $K$ over a finite set $X$ is a $(K+1)$-tuple $ (\rho^0, \ldots, \rho^K)$ of probability distributions over $X$.  We say that $\rho$ has full support if $\cup_{k = 0}^{K} \, \text{supp}\rho^k = X$.
\end{definition}

\begin{definition}\label{def LPS profile}
An LPS profile of order $K$ over $S$ is an $N$-tuple  $\rho \equiv (\rho_1, \ldots, \rho_N)$, where for each player $n$, $\rho_n$ is an LPS over $S_n$ of order $K$. The LPS profile $\rho$ has full support  if each $\rho_n$ has full support, and in this case, let $\ell(\rho) \equiv \max_{n \in \N} \min \{ \, k \mid \cup_{i = 0}^k \, \text{supp} \rho_n^i = S_n \, \}$.
\end{definition}

The next definition gives us a procedure for forming products of the LPSs of the players.

\begin{definition}\label{def LPS beliefs}
Given an LPS profile $\rho$ of order $K$ over $S$, for each player $n$, the induced beliefs $\mu_n$ over $S_{-n}$ of order $K(N-1)$ is defined as follows. For $k = 0, \ldots, K(N-1)$:
\[
\mu_n^k = C_n^k\sum_{{(k_m)}_{m \neq n}}  \otimes_{m \neq n} \rho_{m}^{k_m},
\]
where the sum is over all vectors ${(k_m)}_{m \neq n}$ whose coordinates sum to $k$, and $C_n^k$ is the appropriate normalizing constant that gives us a probability distribution.
\end{definition}

\begin{definition}\label{def LPS best-reply}
Let $\rho$ be an LPS profile of order $K$ and let $\mu_n$ be the induced beliefs for player $n$. For $0 \le k \le K(N-1)$, we say that a strategy $\t_n \in \S_n$ is a best reply of order $k$ against $\rho$ if for all $s_n \in S_n$:
\[
(G_n(\t_n, \mu_n^0), \ldots, G_n(\t_n, \mu_n^k)) \ge_L (G_n(s_n, \mu_n^0), \ldots, G_n(s_n, \mu_n^k)),
\]
where $\ge_L$ is the lexicographic ordering on vectors.
\end{definition}

We are now ready to state the main result of GK.

\begin{theorem}\label{thm GK}
Given a normal-form game $G$, there exist non-negative integers $\ell \le K$ (that depend only on $N$ and the cardinalities of the sets $S_n$) such that a strategy profile $\s$ is a perfect equilibrium of $G$ iff  there exists an LPS profile $\rho$ of order $K$ such that:
\begin{enumerate}
\item $\rho$ has full support and $\ell(\rho) \le \ell$;
\item $\rho_n^0 = \s_n$ for each player $n$;
\item $\s_n$ is a best-reply of order $K$ against $\rho$.
\end{enumerate}
\end{theorem}

For two-player games, we can take $\ell = K = 1$. For the general case, GK only proved that these bounds exist. For our bounds, we need to define a few constants. Let $\kappa = \sum_{n \in \N}|S_n|$; $d = N-1$; $Z = (2d +6)(2d+5)^{\kappa -1}$; $D = 5(\kappa(2d+4) +2)Z$; $\eL = d{(6d - 3)}^{2\k - 1}$.

\begin{theorem}\label{thm main}
 In Theorem \ref{thm GK}, we can take $\ell = 2d(ZD)^2(1+Z)$ and $K = \eL2d(ZD)^2(1+Z)$.
\end{theorem} 

\begin{remark}
We have defined $K$ to be $\ell \eL$, but as the proof of Theorem \ref{thm main} shows, we do not have to let $K$ be this absolute constant, but rather we could let it be $\ell(\rho)\eL$. Put differently, in the statement of Theorem \ref{thm GK}, $\s$ is perfect iff we have an LPS satisfying conditions (1) and (2) and where $\s_n$ is a best-reply of order $\ell(\rho)\eL$ against $\rho$.
\end{remark}

\begin{remark}\label{rem crude bounds}
	Suppose we have an $N$-player game where each player has $a$ actions, then using very crude bounds for $Z$ and $D$, we get $\ell \le 100N^3a^2{(6aN)}^{6aN} \equiv \bar \ell$ and $K \le {(6N)}^{2aN} \bar \ell$. Of course, these bounds are enormously high, but we believe that they can be improved upon (see the last section of the paper).
\end{remark}

\begin{remark}\label{rem integer payoffs}
 As a practical matter, what can be said of our characterization when the payoffs are integers?  It is well-known that even with integer payoffs, a game may not have an equilibrium with rational coordinates---cf. \cite{N1951}. Therefore, the best one can hope for is that the LPS characterization we have involves algebraic numbers, and indeed that is the case.  To see why that is true, observe that an LPS test for perfection involves obtaining probability distributions (the levels of an LPS) that solve a finite system of polynomial equations and inequalities where the coefficients are the payoff numbers in the game. Thus, all the probability distributions of the LPS have only algebraic numbers.  
\end{remark}

\section{Some Facts About Polynomials}
GK prove their theorem by first deriving an equivalence between the statement using LPSs and another involving polynomials.  Applying tools from semi-algebraic geometry for polynomials, they then derive their bounds. The proof of our theorem works with the same set of polynomials and therefore we now review some facts about polynomials and also the above-mentioned equivalence of GK.

A monomial function $F: \Re^{\k} \to \Re$ is of the form $x_1^{i_1}\cdots x_k^{i_k}$ where $i_1, \ldots i_k$ are non-negative integers; its degree in variable $x_l$, denoted $\text{deg}_{x_l} (F)$, is $i_l$ and its total degree, denoted $\text{deg} (F)$, is $i_1 + \cdots + i_k$.  A polynomial is a finite linear combination $\sum_j a_j F_j$ of mononomials $F_j$; its degree in variable $x_l$ is $\max \, \{\, \text{deg}_{x_l} \, F_j \mid a_j  \neq 0 \, \}$ and its total degree is $\max \, \{\, \text{deg} \, F_j \mid a_j  \neq 0 \, \}$, with the degree of the zero function being zero. If  $F: \Re^\k \to \Re^l, l \ge 2$, is a function where each coordinate $F_j$ is a polynomial, the degree of $F$ is the maximum over $j$ of the degree of $F_j$. We call $F$ a polynomial map. 

For a polynomial (or more generally a power series) $f(t) \equiv \sum_k a_k t^k$ of a single variable $t$, the order of $f$, denoted $o(f)$, is the smallest $k$ for which $a_k \neq 0$---the order of the zero function is $\infty$. We say that $f > 0$ (resp. $f \ge 0)$ if $a_{o(f)} > 0$ (resp. either $a_{o(f)} > 0$ or $f \equiv 0$).  For a polynomial map $f: \Re \to \Re^l$, the order of $f$ is $\max_j o(f_j)$.

Now we turn to the description of perfect equilibria using polynomials.
A polynomial strategy-profile is a polynomial map $\eta: \Re \to \prod_{n \in \N} \Re^{S_n}$. For each $n$, the payoff function $G_n$ can be extended uniquely to a multilinear function over the whole of $\prod_{m \in \N} \Re^{S_m}$, still denoted $G_n$.  Given a polynomial strategy profile,  we can now compute the ``payoff'' $G_n(\eta) \equiv G_n \circ \eta$, which is a polynomial whose degree is at most $\sum_{m}\max_{s_m} \text{deg} \, \eta_{m,s_m}$.  We say that a strategy $\t_n \in \S_n$ is a best-reply of order $r$ against a polynomial strategy profile $\eta$ if for each $s_n \in S_n$, $G_n(\t_n, \eta_{-n}) - G_n(s_n, \eta_{-n})$ is either non-negative or of order at least  $r+1$. The following lemma is from GK (see their Claim 3.3).

\begin{lemma}\label{lem polynomial}
Let $\ell \le K$ be non-negative integers and let $\s \in \S$. The following statements are equivalent:
\begin{enumerate}
\item  There exists an LPS profile $\rho$ of order $K$ such that:
\begin{enumerate}
\item $\rho$ has full support and $\ell(\rho) \le \ell$;
\item $\rho_n^0 = \s_n$ for each player $n$;
\item $\s_n$ is a best-reply of order $K$ against $\rho$.
\end{enumerate}
\item There exists a polynomial map $\eta: \Re \to \Re^\k$ such that:

\begin{enumerate}
\item $\eta_{n,s_n} > 0$ for each $n, s_n$ and $o(\eta) \le \ell$;
\item $\eta(0) = \s$;
\item For each $n$, $\s_n$ is a best reply of order $K$ against $\eta$.
\end{enumerate}
\end{enumerate}
\end{lemma}

\begin{remark}
	As the proof of Lemma \ref{lem polynomial} in GK shows, we can take $o(\eta)$ to be equal to $\ell(\rho)$ in going from an LPS profile to the associated polynomial and vice versa, thus giving us an intimate connection between these two ways of looking at perfection. 
\end{remark}

\begin{remark} There are cases where simple bounds can be obtained for $\ell$ and $K$ using the equivalence of the lemma above. First note that $\s$ is a perfect equilibrium if and only if $\s$ belongs to the closure of  $P \equiv \{\,  \tau \in \Re^{\kappa} \mid \tau \in \text{int}(\S), \forall n \in \N, s_n \in S_n, G_n(\s_n, \tau_{-n}) - G_n (s_n, \tau_{-n}) \ge 0 \, \}$. If $P$ is a convex set---as it is the case with two-player or polymatrix games, for example---then consider $\tau^* \in P$. It follows that the linear map $t \mapsto (1-t)\s + t\tau^*$ satisfies conditions 2(a), 2(b) and 2(c) which implies that $\ell$ and $K$ can be taken equal to $1$.   \end{remark}


We conclude this section with the two key results from Real Algebraic Geometry that we referred to before.   Let $F: \Re^\k \to \Re^l$ be a polynomial map of degree $d$ and let $V(F)$ be the set of zeros of $F$. Fix $x \in V(F)$.   The \L{}ojasiewicz inequality provides a lower bound on the value of $F$ in a neighborhood of $x$. Specifically,  there exist positive constants $C, \e, r$ such that $\Vert F(y) \Vert  \ge C \text{dist}{(y, V(F))}^r$ for $y$ such that $\Vert x - y \Vert < \e$, where $\Vert \cdot \Vert$ is the Euclidean norm and $\text{dist}(y, V(F))$ is the Euclidean distance of $y$ to $V(F)$. The smallest $r$ satisfying the inequality is called the \L{}ojasiewicz exponent. Kurdyka and Spordieza \cite{KS2014} show that if $F$ is a polynomial of degree $d \ge 2$, then the \L{}ojasiewicz exponent is $\le d{(6d-3)}^{\k - 1}$. We exploit this estimate in our theorem. The next result concerns the curve-selection lemma.

\begin{proposition}\label{selection}
Let $P$ be a semi-algebraic subset of $\Re^{\kappa}$ defined by polynomials whose total degrees are bounded by $d$. Let $x$ belong to the closure of $P$. There exist $\e>0$ and an analytic function $\phi: [0,\e) \to \Re^{\kappa}$ such that: (1) $\phi((0,\e)) \subset P$; (2) $\phi(0) = x$; (3) $o(\phi) \le 2dZ^2D^2(1+Z)$.
\end{proposition}

\begin{proof}The Quantitative Curve Selection Lemma (Theorem 2 in Basu and Roy \cite{BR2018}) implies that there exist: (1) a semi-algebraic path $\phi: [0, t_0) \to \Re^{\kappa}$; (2) a set of polynomials $f(T, U),g_0(T, U),...,g_{\kappa}(T, U)$ in two variables $(T, U)$; (3) a semi-algebraic function $u:[0, t_0) \to \Re$  such that: (a) $\phi(0) = x$ and $\phi(t) \in P$ for all $t > 0$; (b) $f(t, u(t)) = 0$ for all $t \in (0, t_0)$; (c) $\phi(t) = (\frac{g_1(t,u(t))}{g_0(t,u(t))},...,\frac{g_{\kappa}(t,u(t))}{g_{0}(t,u(t))})$, for $t>0$.\footnote{Since the \textit{description of $x$} (see Basu and Roy \cite{BR2018} for a definition) uses univariate polynomials with coefficients in $\Re$, $x$ admits a trivial description $h = (X, X-1, X-x_1,...,X-x_{\kappa})$. Therefore, the coefficients of the polynomials used in the \textit{description of the semi-algebraic path} $\phi$ can be taken directly as $\Re$.} Moreover, $$\max\{\text{deg}_{T}(f), deg_{T}(g_0),..., deg_{T}(g_{\kappa})\} = 2dZD^{2}$$ and  $$\max\{\text{deg}_{U}(f), deg_{U}(g_0),..., deg_{U}(g_{\kappa})\} = Z.$$ Viewing the polynomial $f$ as a polynomial with complex coefficients, it follows from the algebraic closure of complex Puiseux series that the root $u(\cdot)$ can be assumed to be a real Puiseux series. Moreover, by Riemman's method of resolution of singularities (see Theorem 1.5 in Koll\'ar \cite{K2007}) we have that $u(t)$ is a real Puiseux series with nonnegative exponents: $u(t) = \sum_{k \ge 0}a_k t^{k/q}$, where $q \le Z$. Riemann's result also implies that the Puiseux series converges in a neighborhood of zero. Let $\xi$ be the order of $u$. Now, a necessary condition for $(t, u(t))$ to be a root of $f(T,U) = \sum_{(i,j)} c_{ij}T^i U^j$ for $t \in [0,t_0)$ is that the lowest powers of $t$ after substituting $(t, u(t))$ for $(T,U)$ must cancel. Therefore, there must be at least two monomials $c_{ij} T^i U^j$ and $c_{i'j'}T^{i'} U^{j'}$ such that both give the same degree $\beta$ on $t$ after substituting $(t,u(t))$ in the monomials, and other monomials give degrees $\ge \beta$. Therefore, $i + j\xi = i' + j' \xi \iff \xi = \frac{i - i'}{j' - j} \le 2dZD^2$. 

Since $deg_T(g_{i})$ is bounded by $2dZD^2$ and $deg_{U}(g_{i})$ is bounded by $Z$, it follows that the order of $g_{i}(t,u(t))$ is bounded by $(2dZD^2 + 2d(ZD)^2)$. Changing variables from $t$ to $t^q$, it follows that the order of  $g_{i}(t^q,u(t^q))$ is less than $(2dZD^2 + 2d(ZD)^2)Z$. Since $\phi$ is continuous at $0$, it follows that $\frac{g_{i}(t^q, u(t^q))}{g_{0}(t^q, u(t^q))}, i = 1,...,\kappa$ is a power series with order bounded by $(2dZD^2 + 2d(ZD)^2)Z = (2d(ZD)^2)(1+Z)$.\end{proof}

\section{Proof of Theorem \ref{thm main}}
As we remarked in Section 2, for the case $N = 2$,  we can take $\ell = K = 1$. Therefore, assume $N > 2$.  We prove the theorem by invoking the equivalence in Lemma \ref{lem polynomial}.  Fix $\s \in \S$. Suppose $\ell$ and $K$ are as specified in Theorem \ref{thm main} and suppose $\eta: \Re \to \Re^k$ is a polynomial map satisfying  properties 2(a)-2(c) of Lemma \ref{lem polynomial}.  We show that $\s$ is perfect.  There is nothing to prove if $\s$ is completely mixed; therefore assume that it is not. For each $n$, let $T_n$ be the set of pure strategies that are best replies against $\s$. Let $T_n^1$ be the subset of $T_n$ consisting of those strategies that are best replies of order $K$ against $\eta$; then $T_n^1$ includes the support of $\s_n$. Let $T_n^0 = T_n \backslash T_n^1$. Let $F: \Re^\k \times \prod_n \Re^{T_n^0} \to \prod_{n} \Re^{T_n}$ be the polynomial whose coordinate for  $(n, t_n)$ is  $F_{n, t_n}(\t, c) \equiv G_n(t_n, \tau_{-n}) - G_n(\s_n, \tau_{-n})$ if $t_n \in T_n^1$ and $F_{n, t_n}(\t, c) \equiv G_n(t_n, \tau_{-n}) - G_n(\s_n, \tau_{-n}) + c_{n,t_n}^2$ if $t_n \in T_n^0$. (In case $T_n^0$ is empty for some $n$, then we do not have the coordinate $c$ for him; in particular if $T_n^0$ is empty for all $n$, $F$ is a function defined on $\Re^\k$.)
For all small $t$, each player $n$ and each strategy $t_n \in T_n^0$, $G_n(t_n, \eta_{-n}(t)) - G_n(\s_n, \eta_{-n}(t)) <0$. Therefore, the function $c_{n,t_n}(t) \equiv  {(G_n(\s_n, \eta_{-n}(t) - G_n(t_n, \eta_{-n}(t)))}^{1/2}$ is well-defined.  The map $F(\eta(t), c(t))$ now has order $K + 1$ as a function of $t$.

Let $W$ be the set of $(\t, c) \in \Re^{\kappa} \times \prod_n \Re^{T_n^0}$ such that $\t_{n,s_n} = 0$ for some $n$ and $s_n$ that is not in the support of $\s_n$.  We claim that $(V(F) \backslash W) \cap U$ is nonempty for each small neighborhood $U$ of $\s$. Indeed, by the \L ojasiewicz inequality, if this intersection is empty, then for each small $t$, since $d((\eta(t), c(t)), V(F)) = O(t^{o(\eta)})$, the order of $F(\eta(t), c(t))$ is no more than $o(\eta)\eL \leq K$, which is a contradiction with the conclusion from the previous paragraph.  Hence, the intersection is nonempty.   

Take now a sequence of neighborhoods $U_k$ of $\s$ whose intersection is $\s$.  For each $k$, pick a point $(\s^k, c^k) \in (V(F) \backslash W) \cap U_k$.  For each $n$, all the strategies in $T_n^1$ are equally good replies against $\s^k$, and at least weakly better than those in $T_n^0$; as strategies in $S_n \backslash T_n$ are inferior replies against $\s$, they remain so against $\s^k$ for large $k$. Therefore, $\s$ is a perfect equilibrium.

To prove the other direction, let now $\s$ be a perfect equilibrium of $G$.  Let $P$ be the set of completely mixed strategy-profiles $\t \in \S$ such that  for all $n \in \N, s_n \in S_n$, $G_n(\s_n, \tau_{-n}) - G_n (s_n, \tau_{-n}) \ge 0$. The set $P$ is semi-algebraic. Moreover, as $\s$ is perfect, $P$ is  non-empty and $\s$ belongs to the closure of $P$. By Proposition \ref{selection}, it follows that there exists an analytic function $\varphi: [0,\e) \to \Re^\kappa$ such that $\varphi((0,\e)) \subset P$, $\varphi(0) = \s$ and $o(\varphi) \le \ell$. Therefore, $K \le (2d(ZD)^2)(1+Z)\eL$. Now consider the polynomial $\eta$ of order $K$ obtained by truncating each coordinate of $\varphi$ to its first $K+1$ terms. This polynomial satisfies conditions 2(a)-2(c) of Lemma \ref{lem polynomial} and our theorem is proved.

\section{Concluding Remarks}
The two quantitative results concerning the \L{}ojasiewicz inequality and the curve-selection lemma that we invoke hold for the case of arbitrary polynomials.  In the context of game theory, the polynomials we are considering are very special: they happen to be multilinear functions.  Therefore, it is worthwhile investigating whether the bounds in this paper can be tightened. 

The idea of trembles in the definition of perfection is the basis for a number of refinements, and we can obtain a finite characterization for several of those as well.  As extensive-form perfection is the same as normal-form perfection applied to the agent-normal-form, the results here extend immediately to it.  As for properness \cite{M1978}, as GK show, there is an equivalent definition using LPSs that is similar to that for perfection with the added restriction that superior replies are infinitely more likely in the LPS.  By incorporating an additional variable $\e$ (used in the definition of $\e$-properness), we get a finite characterization where the variables $\k$ and $d$ are augmented by one.  Finally, a more challenging and important open problem is to obtain a similar characterization for stable sets \cite{KM1986}.\footnote{The definition of Kohlberg-Mertens stability invokes a minimality property and therefore these sets are not semi-algebraic.  But components of perfect equilibria that satisfy their robustness property are semi-algebraic, and it is these sets that we could hope to identify in a finite way.}

\end{document}